\documentclass[letterpaper,12pt]{article}
\usepackage[margin=1.00in]{geometry}

\usepackage{endnotes, amsmath, amssymb, amsfonts, amsthm}
\usepackage{comment, url, graphicx, color}
\usepackage{setspace}
\usepackage{framed, enumerate}
\newcommand{\E}           {\mathbb{E}}        % expected value
\newcommand{\bF}          {\mathbb{F}}        % filtration
\newcommand{\bP}          {\mathbb{P}}        % probabiliy
\newcommand{\R}           {\mathbb{R}}        % real numbers
\newcommand{\sA}          {\mathcal{A}}
\newcommand{\sF}          {\mathcal{F}}

\newcommand{\sN}          {\mathcal{N}}

\newcommand{\sU}          {\mathcal{U}}
\newcommand{\rcmp}       {r^{comp}}

\numberwithin{equation}{section}
\theoremstyle{plain}                % title and  number in bold, text italic
\newtheorem{theorem}{Theorem}[section]

\newtheorem{proposition}[theorem]{Proposition}
\newtheorem{corollary}[theorem]{Corollary}
\theoremstyle{definition}           % title and number in bold, text normal
\newtheorem{definition}[theorem]{Definition}

\newtheorem{tablebox}[theorem]{Table}

\theoremstyle{remark}               % title and number in italic, text normal
\newtheorem{remark}{Remark}[section]

\begin{document}

% title
\begin{center}
%  \vskip .4in
  {\Large Existence of a Radner equilibrium in a model with 
  transaction costs}
  \vskip .3in
  Kim Weston\footnote{The author would like to thank the anonymous referee and editor, Frank Riedel, as well as Gordan \v{Z}itkovi\'c, Mihai S\^irbu, Kasper Larsen, Johannes Muhle-Karbe, Frank Seifried, and Matteo Burzoni for helpful discussions on this work.  The author acknowledges support by the National Science Foundation under Grant No.\ DMS-1606253.  Any opinions, findings and conclusions or recommendations expressed in this material are those of the author and do not necessarily reflect the views of the National Science Foundation (NSF).} \\
  Rutgers University \\
  Department of Mathematics \\
  Piscataway, NJ  08854, USA
  \vskip .3in
  \today
\end{center}

\vskip .4in

\abstract{We prove the existence of a Radner equilibrium in a model with proportional transaction costs on an infinite time horizon and analyze the effect of transaction costs on the endogenously determined interest rate.  Two agents receive exogenous, unspanned income and choose between consumption and investing into an annuity.  After establishing the existence of a discrete-time equilibrium, we show that the discrete-time equilibrium converges to a continuous-time equilibrium model.  The continuous-time equilibrium provides an explicit formula for the equilibrium interest rate in terms of the transaction cost parameter. We analyze the impact of transaction costs on the equilibrium interest rate and welfare levels.}


\section{Introduction}
%The individual agent's optimal investment problem in the presence of proportional transaction costs is well-studied, dating back to \ref{MC76JET} \ref{DN90MOR}.  
We study an incomplete Radner equilibrium with proportional transaction costs.  Transaction costs influence asset prices, yet most asset pricing models with transaction costs take some or all asset prices as given.  We seek to answer two questions:
\begin{enumerate}[\itshape(1)]
  \item \emph{%How are interest rates formed in a finite-agent equilibrium?  }
  Does a finite-agent equilibrium with transaction costs and an endogenously determined interest rate exist?}
  \item \emph{If an equilibrium exists, what is the effect of transaction costs on the interest rate and welfare levels?}
\end{enumerate}
We devise a model to answer these questions and reveal an unexpected %surprising 
effect from transaction costs on the endogenous asset prices.  Depending on agent risk preferences and consumption smoothing over time, a proportional transaction cost can increase or decrease equilibrium interest rates.  We provide an explicit formula for the continuous-time equilibrium interest rate in terms of the transaction costs and other input parameters.  We find that welfare decreases with increases in the transaction costs.

% previous work deficiencies
Proving the existence of a general equilibrium is a challenging problem on its own, and frictions exacerbate the difficulties. Existing work on equilibria with transaction costs lacks the ability to endogenously derive asset prices while providing rigorous justification in a finite-agent Radner equilibrium.  The works \cite{V98RFS}, \cite{LMW04JPE}, and \cite{D16wp} rely on an exogenously specified bank account, while %\cite{HL92JEDC},
\cite{LMW04JPE} and \cite{BUV14wp} provide numerics but no existence result.  Models with a continuum of agents are introduced in \cite{VV99ET}, \cite{V98RFS}, and \cite{H03JET}.  An approximate equilibrium concept is introduced in \cite{HMK17wp}.  In contrast to the existing literature, we introduce transaction costs into a finite-agent Radner equilibrium with unspanned income and consumption over an infinite time horizon.  This set-up is longstanding without transaction costs; see, for instance, \cite{C01JET} and \cite{W03AER}.  We prove the existence of a proportional transaction cost Radner equilibrium and analyze its effects on endogenous asset prices.% for two exponential agents with unspanned income, where the traded security price is determined endogenously.

In our model, two exponential investors receive unspanned income, consume, and trade in an annuity market on an infinite time horizon.  The individual agent's optimal investment problem in the presence of proportional transaction costs is well-studied and dates back to \cite{MC76JET} and \cite{DN90MOR}.  The equilibrium setting in \cite{VV99ET} is most similar to ours: In \cite{VV99ET}, a continuum of agents trade in two annuities in an overlapping generations equilibrium while one of the annuities faces real proportional transaction costs from trading.  However, \cite{VV99ET} assume that the zero transaction cost economy supports an equilibrium in which the agents' wealth first increases then decreases.  We make no such assumption, and our model does not exhibit this behavior.  We provide an example of an equilibrium with proportional transaction costs, which derives all traded asset prices endogenously and analytically describes the effect of transaction costs on the equilibrium interest rate.

% discuss the lambda-->0 limit and literature?
Understanding the effects of transaction costs analytically is often not possible.  In \cite{JS04fS}, \cite{ST13SICON}, and \cite{KMK16MF}, the authors derive a Taylor expansion for the single-agent value function and no-trade boundaries for small transaction costs.  In equilibrium, we seek to understand the effect of transaction costs on the agents' behavior and equilibrium interest rate.  The closeness of the agents' input parameters determine explicitly when the agents are motivated to trade in equilibrium.    We also derive an explicit formula for the interest rate in terms of the transaction costs in continuous time.  For agents $i=1,2$, we let $\alpha_i>0$ be agent $i$'s risk aversion, $\beta_i>0$ be the time-preference parameter for consumption, $\mu_i$ be agent $i$'s income stream's drift, and $\sigma_i$ be the volatility of the (unspanned) Brownian component in agent $i$'s income stream. We let $\lambda\in[0,1)$ be the proportional transaction cost parameter.  Then in an equilibrium in which agent 1 chooses to buy and agent 2 chooses to sell, the equilibrium interest rate is given explicitly by
$$
  r(\lambda) = \frac{\tilde\beta_1/\alpha_1+\tilde\beta_2/\alpha_2}
    {\frac{1}{\alpha_1(1+\lambda)}+\frac{1}{\alpha_2(1-\lambda)}}.
$$

\label{intro:formula}The relationship between the agents' risk aversions determine whether equilibrium interest rates will increase or decrease when transaction costs are introduced.  Our model shows that typically the presence of transaction costs will decrease the equilibrium interest rate.  However, when one agent seeks to trade aggressively in order to ensure her future consumption while the other agent trades timidly and consumes more readily, then it is possible for small transaction costs to increase the equilibrium interest rate.  In this case, the agents balance each other in their consumption timing and risk appetites so that a premium is placed on the annuity's future consumption stream when transaction costs are introduced.  Our model is admittedly stylized, yet it captures the equilibrium interest rate and trading behavior with transaction costs in a long-standing, classical setting.

% shadow price approach + literature
We employ a shadow price approach to establish an equilibrium in order to gain tractability of the single-agent problems.  Shadow prices represent the traded asset price in a least-favorable frictionless market completion, where the optimal investment and consumption strategies align between the frictionless shadow market and the transaction cost market.  Shadow prices for proportional transaction costs were introduced by \cite{JK95JET} and \cite{CK96MF} and have since been established in increasingly greater generality; see, for example,  \cite{KMK10AAP} and \cite{CS16AAP}.  Because least-favorability is investor specific, each economic agent will select her own frictionless shadow market to perform utility maximization.  We link the investor-specific shadow markets using a ``closeness'' condition in equilibrium.  %While shadow price techniques are traditionally used on the dual side of the problem, here, we study the primal shadow markets directly.
We show that a unique equilibrium asset price is only guaranteed when a trade occurs.  Otherwise, agents' shadow prices allow for a range of prices consistent with the equilibrium.

% important components + income literature + annuity literature
Several components of this equilibrium example are crucial for obtaining our results.  We rely on the agents' exponential preferences and income processes with independent increments for tractability of the single-agent problem similar to \cite{W03AER}, \cite{CLM12JET}, and \cite{LS16MAFE}.  In a frictionless model with deterministic interest rates, an annuity is spanned by a bank account, and vice versa.  With transaction costs, we cannot freely move between an annuity and the bank account as the traded security.  We work with the annuity as the traded security similar to \cite{H03JET} and \cite{VV99ET}.  This choice yields trading strategies in which the agents choose to do the same thing at every time point: either buy, sell, or trade nothing.  Theorem~\ref{thm:bank} proves that the constant interest rate equilibrium obtained by trading in the annuity is not possible when the bank account is the traded security.  This model is indeed highly stylized, yet it provides the first existence proof of a finite-agent Radner equilibrium with proportional transaction costs in which all traded asset prices are derived endogenously.

% organization
The paper is organized as follows.  Section~\ref{section:discrete} describes the discrete-time equilibrium and proves its existence in Theorem~\ref{thm:discrete_equilibrium}.  Section~\ref{section:cts} considers a continuous-time equilibrium model.  The existence of an equilibrium is established in Theorem~\ref{thm:cts_eq}, and it is shown to be the limit of discrete-time equilibria.  We analyze the impact of transaction costs on interest rates and welfare in Section~\ref{section:effects}.  Section~\ref{section:bank} discusses a transaction cost equilibrium with a traded bank account.  The proofs are contained in Section~\ref{section:proofs}.

% income
% [1/13/17, 2:24:46 PM] Kasper Larsen: .	N. Wang, Caballero meets Bewley: The permanent-income hypothesis in general equilibrium, Amer. Econ. Rev. 93 ?%(2003) 927Ð936.
%[1/13/17, 2:25:51 PM] Kasper Larsen: .	ÊL.E. Calvet, Incomplete markets and volatility, J. Econ. Theory 98 (2001) 295Ð338.
\section{Discrete-Time Equilibrium}\label{section:discrete}

We consider a discrete-time infinite time horizon Radner equilibrium without a risky asset.  There is a single consumption good, which we take to be the numeraire.  Time is divided into intervals $[t_n,t_{n+1})$, $n\geq 0$, where $t_n:=n \Delta$ and $\Delta>0$.  An annuity, denoted by $A$, is in one-net supply and is available to trade with an exogenously specified proportional transaction cost $\lambda\in[0,1)$.  One share in the annuity delivers consumption units at a rate of one per unit time over all future time intervals.  Thus, a share in the annuity will deliver $\Delta$ consumption units over each time interval $[t_n,t_{n+1})$.

The risk-free rate $r>0$ will be determined endogenously in equilibrium using the equilibrium annuity values.  We will focus on equilibria allowing for constant, positive interest rates even though our definition of equilibrium does not exclude general interest rates.  
The annuity dynamics are given by
$$
  A_{t_{n+1}} - A_{t_n} = \left(A_{t_n} r - 1\right)\Delta, 
  \ \ \  A_0>0.
$$
In this case, the annuity value will be the constant $A_{t_n} = A= 1/r$.  \label{comment:bank}We choose to study a model with a traded annuity rather than a bank account because the annuity will provide the mathematical structure needed to establish an equilibrium with transaction costs.  We discuss the implications of a traded bank account in Section~\ref{section:bank}.

Each agent has the exogenous income stream $Y_i = (Y_{i t_n})_{n\geq 0}$ given by,
$$
  Y_{i t_{n+1}} = Y_{i t_n}+ \mu_i\Delta+\sqrt{\Delta}\sigma_i Z_{i t_{n+1}}, \ \ \ Y_{i0}\in\R,
$$
where $\mu_i\in\R$, $\sigma_i>0$, and $Z_{i t_{n+1}}\sim\sN(0,1)$ for $i=1,2$.  The random variables $(Z_{i t_n})_{n\geq 1}$ are independent, and $Z_{1 \cdot}$ and $Z_{2 \cdot}$ are possibly correlated.  The agents are also endowed with an initial allocation of annuity shares $\theta_{i 0}\in\R$ such that $\theta_{1 0}+\theta_{2 0}=1$.

The flow of information in this economy is given by $\bF=(\sF_{t_n})_{n\geq 0}$, where $\sF_{t_n}=\sigma(Z_{i t_1},\cdots,Z_{i t_n}: i=1,2)$.  All processes are assumed to be adapted to $\bF$, and all agents share the same filtration and probability $\bP$.  All equalities are assumed to hold $\bP$-almost surely.

\subsection{Individual Agent Problems}

Rather than deal directly with an optimization problem in a market with frictions, we cast each individual investor's problem as a problem in her own frictionless shadow market.  In equilibrium, the agents' shadow markets will be related, and a unique (non-shadow) equilibrium price for the traded annuity will exist when a trade occurs.  Yet the individual optimization problems are treated in isolation as frictionless.  Therefore, only in the next section (Section~\ref{section:discrete_equilibrium}) will the parameter $\lambda$ appear.

%Consider the single-agent investment and consumption problem, where the subscript $i$'s are dropped for notational convenience.  
We first consider the single-agent investment and consumption problem for agent $i\in\{1,2\}$.  At time $t_n\geq 0$, agent $i$ chooses to consume $c_{t_n}$ units of the consumption good and invest $\theta_{t_n}$ shares in the annuity beginning with an initial allocation of $\theta_{0}$.  We consider equilibria for which the value of the shadow annuity and shadow interest rate (to be determined endogenously in equilibrium) are constants $A_{i t_n}=A_i=1/r_i$ and $r_i>0$, respectively.  %Consumption occurs at the rate $c=(c_{t_n})_{n\geq 0}$.

For a given investment strategy $\theta$, agent $i$'s shadow wealth is defined by
$$
  X_{i t_n} := \theta_{t_{n}} A_{i t_n},
$$
with the self-financing condition
\begin{equation}\label{eqn:discrete_sf}
  (\theta_{t_{n+1}}-\theta_{t_{n}})A_{i t_{n+1}} 
  = \left(Y_{i t_n}-c_{t_n}+\theta_{t_n}\right)\Delta, \ \ \ n\geq 0.
\end{equation}
For a given consumption and investment strategy $(c,\theta)$, % with $\theta_0=\theta_{i0}$, 
the wealth evolves as
$$
  X^c_{i t_{n+1}}-X^c_{i t_n} 
  = \left(X^c_{i t_n}r_i + Y_{i t_n} - c_{t_n}\right)\Delta, 
  \ \ \  X_{i0} = \theta_{0}A_{i0}.
$$
%where we necessarily have that $\theta_{t_n} = $. %XXX fill me in!
Given a consumption strategy $c$ and an initial share allocation $\theta_{0}$, the self-financing condition dictates the investment strategy $\theta$.

We consider agents with exponential preferences over running consumption; that is, agent $i$'s utility function is $c\mapsto -e^{-\alpha_i c}$ for $\alpha_i>0$.  The agents prefer consumption now to consumption later, which is measured by their time-preference parameters $\beta_i>0$.  For $i=1,2$, $t\geq 0$, and $c\in\R$, we define $U_i(t,c) :=  -e^{-\beta_i t-\alpha_i c}$.

\begin{definition}\label{def:discrete_admissibility}
A consumption strategy $c$ is called \textit{admissible for agent $i$} if $c$ 
%$\sum_{n=0}^\infty \E\left[e^{-\alpha_i c_{t_n}}\right]<\infty$ and 
satisfies the transversality requirement
$$
  \E\left[\exp\left(-\beta_i t_n - \alpha_i r_i X^c_{t_n} - \alpha_i Y_{t_n}\right)\right] \longrightarrow 0 \ \ \ \text{ as $n\rightarrow\infty$.}
$$
In this case, we write $c\in\sA^\Delta_{i}$.
\end{definition}

The value function is defined by
\begin{equation}\label{eqn:discrete_value_function}
  V^\Delta_i (x,y) := \sup_{c\in\sA^\Delta_i} \sum_{n=0}^\infty 
  \E\left[U_i(t_n, c_{t_n})\right], 
  \ \ \ x,y\in\R.
\end{equation}

\label{remark:constant-r}
  It is possible to consider non-constant shadow interest rates in the individual agent formulation, however it comes at the cost of not being able to explicitly describe the individual agents' transversality conditions.  We refer the reader to Chapter 9 Section D of \cite{D01} for further details.  Our constant shadow interest rate formulation allows us to explicitly characterize the optimal consumption and wealth processes, which we summarize in Theorem~\ref{thm:discrete_individual}.

\begin{theorem}\label{thm:discrete_individual}
  Agent $i$'s optimal consumption and wealth process in \eqref{eqn:discrete_value_function} are given by
  % XXX check the Delta's in c, X, and J!
  \begin{equation}\label{eqn:discrete_opt_consumption}
    \hat c_{i t_n} = r_i\hat X_{i t_n}+ Y_{i t_n}
    + \frac{1}{\alpha_ir_i\Delta}\left(\tilde\beta_i\Delta-\log(1+r_i\Delta)\right)
  \end{equation}
  and
  \begin{equation}\label{eqn:discrete_opt_wealth}
    \hat X_{i t_n} = \frac{\theta_{i0}}{r_i} + \frac{t_n}{\alpha_i r_i}
    \left(\frac{1}{\Delta}\log\left(1+r_i\Delta\right)
    -\tilde\beta_i\right),
  \end{equation}
  where $\tilde\beta_i := \beta_i+\alpha_i\mu_i 
  -\frac{\alpha_i^2\sigma_i^2}{2}$. 
  Moreoever, the value function can be expressed in the form
  $$
    V^\Delta_i(x,y) = J^\Delta_i(x,y) 
    := - \frac{1}{r_i\Delta} \left(1+r_i\Delta\right)^{1+\frac{1}{r_i\Delta}}\ 
    \exp\left(-\alpha_i r_i %\Delta \cdot 
    x - \alpha_i y 
    - \frac{\tilde\beta_i}{r_i}\right).
  $$
\end{theorem}

\subsection{Equilibrium}\label{section:discrete_equilibrium}
%When no trade occurs, the definition of equilibrium must allow for the possibility of a range of annuity values consistent with equilibrium.  Shadow prices allow us to relate both agents' willingness to trade, even when no trade occurs.  They also enable us to compute the range of annuity values consistent with an equilibrium.

The definition of equilibrium must allow us to relate both agents' willingness to trade, even when no trade occurs due to frictions.  Shadow prices provide us with this mechanism and a way to compute the range of annuity values consistent with an equilibrium.  
Typically, shadow prices are used as a tool to establish properties of an original model with frictions; see, for example, \cite{KMK10AAP} and \cite{CS16AAP}.  Here, we work with shadow prices directly, and we subsequently determine transaction cost models consistent with our agents' shadow markets.

\begin{definition}\label{def:discrete_equilibrium}
For the transaction cost parameter $\lambda\in[0,1)$, an \textit{equilibrium with transaction costs} is given by a collection of processes $(A_i, \hat c_i, \hat\theta_i)_{i=1,2}$ such that
\begin{enumerate}[(i)]
  \item \label{def:discrete_clearing}
  Real and financial markets clear for each $n\geq 0$:
    $$
      \sum_{i=1}^2\hat c_{i t_n}\Delta = \Delta +\sum_{i=1}^2 Y_{i t_n}\Delta 
      - 2\lambda\left|\hat\theta_{1 t_{n+1}}-\hat\theta_{1 t_n}\right|A_{t_{n+1}} 
      \ \ \ \text{ and } \ \ \ 
      \hat\theta_{1 t_n} + \hat\theta_{2 t_n} = 1,
    $$
    where in the event of a trade, we define 
    $A_{t_{n+1}}:=\frac{A_{i t_{n+1}}}{1+\lambda}$ 
    if agent $i\in\{1,2\}$ purchases a positive number of annuity shares; 
    that is, $\theta_{i t_{n+1}}-\theta_{i t_n}>0$.
  \item For each agent $i=1,2$, the consumption and investment strategies, 
    $\hat c_i$ and $\hat\theta_i$ with $\hat\theta_{i0}=\theta_{i0}$, are 
    optimal with the shadow annuity price $A_i$:
    $$ 
      V^\Delta_i(\theta_{i0}A_{i0},Y_{i0}) = \sum_{n=0}^\infty \E
      \left[U_i( t_n, \hat c_{i t_n})\right].
    $$
  \item \label{def:discrete_closeness} 
  The shadow markets remain ``close enough'' to the underlying transaction 
  cost market in the following sense:  For each $n\geq 0$,
  $$
    \frac{A_{1 t_n}}{A_{2 t_n}}
    \in\left[\frac{1-\lambda}{1+\lambda},\frac{1+\lambda}{1-\lambda}\right].
  $$
  Moreover, for $n\geq 1$, if $\hat\theta_{1 t_n}-\hat\theta_{1t_{n-1}}>0$ then 
  $A_{1 t_n} = A_{2 t_n}\cdot\frac{1+\lambda}{1-\lambda}$.  
  If $\hat\theta_{1 t_n}-\hat\theta_{1 t_{n-1}}<0$ then $A_{1 t_n} = A_{2 t_n}\cdot\frac{1-\lambda}{1+\lambda}$. 
\end{enumerate}
\end{definition}

\begin{remark}[On Definition \ref{def:discrete_equilibrium} \eqref{def:discrete_closeness}]\label{remark:closeness}
  %Condition \eqref{def:discrete_closeness} requires the shadow markets 
  %to maintain a ratio relation.  
  Let us consider a single-agent optimization 
  problem for a risky asset with frictions $S$ and a shadow price 
  $\tilde S$.  The shadow price along with the optimal trading strategy 
  $\hat\theta$ will satisfy $\tilde S_{t_n} 
  \in [(1-\lambda)S_{t_n}, (1+\lambda)S_{t_n}]$, 
  $\tilde S_{t_n} = (1-\lambda)S_{t_n}$ when $\hat\theta_{t_n}-\hat\theta_{t_{n-1}}<0$, and 
  $\tilde S_{t_n} = (1+\lambda)S_{t_n}$ when $\hat\theta_{t_n}-\hat\theta_{t_{n-1}}>0$.  
  Condition \eqref{def:discrete_closeness} in Definition~\ref{def:discrete_equilibrium} enforces this relationship between 
  both agents' shadow markets and the underlying market.  \label{remark:closeness_explanation}In the absence of condition \eqref{def:discrete_closeness} and when trade does not occur, there is no connection between the shadow annuity markets and the underlying market.  In this case, there are infinitely many no-trade equilibria, even in the frictionless ($\lambda = 0$) case.
\end{remark}

Since each agent optimizes in her own shadow market while maintaining the ``closeness'' condition \eqref{def:discrete_closeness}, a unique market annuity rate is only guaranteed when trade occurs.  When trade does not occur in a given period, there is a range of possible annuity values (and corresponding interest rates) consistent with equilibrium.  \label{explanation:shadow}In \cite{LMW04JPE, HMK17wp}, the authors have a single equilibrium price rather than two shadow prices, but these works need to allow for the agents to pay different transaction costs based on a total exogenous cost.  We are able to avoid this endogenous splitting of costs while maintaining tractability in part because we choose to work with shadow markets.

The following is the main result of the section.  The proof is in Section \ref{section:proofs}.
% main result in discrete-time
\begin{theorem}\label{thm:discrete_equilibrium}
  %Assume that $\mu_i>\frac{\alpha_i\sigma_i^2}{2}$, and 
  Let $\tilde\beta_i:=\beta_i+\alpha_i\mu_i-\frac{\alpha_i^2\sigma_i^2}{2}$, 
  and assume that $\tilde\beta_i$ is strictly positive for $i=1,2$. 
  For $\lambda\in[0,1)$, there exists an equilibrium with strictly positive 
  constant shadow interest rates $r_1, r_2$ and constant shadow 
  annuity values $A_1=1/r_1$, $A_2=1/r_2$.  The optimal consumption and  
  wealth processes for investor $i=1,2$, are given by 
  \eqref{eqn:discrete_opt_consumption} and \eqref{eqn:discrete_opt_wealth},
  respectively. 
  
  \noindent{\bf Case 1: } A no-trade equilibrium occurs if
  \begin{equation}\label{eqn:discrete_case1}
    \frac{e^{\tilde\beta_2\Delta}-1}{e^{\tilde\beta_1\Delta}-1}\in
    \left[\frac{1-\lambda}{1+\lambda},\frac{1+\lambda}{1-\lambda}\right].
  \end{equation}
  In this case,
  $$
    r_1 = \frac{e^{\tilde\beta_1\Delta}-1}{\Delta} \ \ \ \text{ and } \ \ \ 
    r_2 = \frac{e^{\tilde\beta_2\Delta}-1}{\Delta}.
  $$
  The range of possible constant, non-shadow interest rates that are 
  consistent with this equilibrium is given by $r=(r_{t_n})_{n\geq 0}$ with 
  $$
    r_{t_n} \in\left[\frac{1-\lambda}{\Delta}
    \left(e^{\max(\tilde\beta_1, \tilde\beta_2)\Delta}-1\right), 
    \frac{1+\lambda}{\Delta}\left(e^{\min(\tilde\beta_1,
    \tilde\beta_2)\Delta}-1\right)\right] \neq \emptyset.
  $$
  
  \noindent{\bf Case 2: }
  There exists an equilibrium in which 
  agent $1$ will purchase shares of the annuity 
  in equilibrium at all times $t_n\geq 0$ (while agent $2$ sells shares) if
  \begin{equation}\label{eqn:discrete_case2}
    \frac{e^{\tilde\beta_2\Delta}-1}{e^{\tilde\beta_1\Delta}-1} 
    > \frac{1+\lambda}{1-\lambda},
  \end{equation}
  where the interest rate $r>0$ is uniquely determined by 
  \begin{equation}\label{eqn:discrete_r} % i=1, j=2
    \left(1+\frac{r\Delta}{1-\lambda}\right)^{\frac{1}{\alpha_2\Delta}}
    \left(1+\frac{r\Delta}{1+\lambda}\right)^{\frac{1}{\alpha_1\Delta}} 
    = e^{\frac{\tilde\beta_1}{\alpha_1}
    +\frac{\tilde\beta_2}{\alpha_2}},
  \end{equation}
  and the shadow interest rates are given in terms of
  $$
    r = r_1 (1+\lambda) = r_2 (1-\lambda).
  $$
\end{theorem}
\label{discussion:betas}The parameters $\tilde\beta_1$ and $\tilde\beta_2$ represent the agents' time preference parameters for consumption adjusted for risk and income.  Strictly positive $\tilde\beta_i$ parameters correspond to an economy that allows for strictly positive equilibrium shadow interest rates.  Strictly positive shadow interest rates, in turn, ensure that the shadow annuity values are well-defined.  Allowing even one of the $\tilde\beta_i$ parameters to cross zero would cause the corresponding shadow annuity to be infinitely valued.  Financially, this case corresponds to the case when a (zero transaction cost) shadow annuity with a constant interest rate cannot be replicated by a bank account because of the bank account's dwindling value.

When the agents' parameters $\tilde\beta_1$ and $\tilde\beta_2$ are sufficiently close, as in \eqref{eqn:discrete_case1}, then the agents are not motivated to trade because of the relatively high transaction costs.  Trading occurs in equilibrium only when the parameters $\tilde\beta_1$ and $\tilde\beta_2$ are sufficiently far apart to overcome the transaction costs.  In this case, the agents' strategies are very simple:  either buy or sell the exact same amount at every time period.
\begin{remark}
  If the inequality \eqref{eqn:discrete_case2} is flipped so that
  $$
    \frac{e^{\tilde\beta_1\Delta}-1}{e^{\tilde\beta_2\Delta}-1} 
    > \frac{1+\lambda}{1-\lambda},
  $$
  then we can conclude an analogous result in which the roles of agent 
  $1$ and $2$ are interchanged.
\end{remark}

\section{Continuous-Time Equilibrium}\label{section:cts}
%Establishing the existence of an equilibrium in discrete time was the first step to determining the form of a continuous-time transaction cost equilibrium.  %When transaction costs appear in the real goods clearing condition, it is a priori unclear how to handle the clearing condition's passage to the continuous-time limit because trading strategies in continuous time are often not required to be smooth in time, whereas the other real goods clearing terms arise from a running consumption optimization framework and are already expressed as rates.  Typical single-agent utility maximization problems with transaction costs allow for finite variation strategies and optimally trade a risky security on a (Brownian) local time scale; see, for example, \cite{SS94AAP}.  Our continuous-time equilibrium definition below will require the agents' trading strategies to be differentiable in time, contrary to the single agent utility maximization and equilibrium literature \cite{SS94AAP, CS16AAP, KLLS91SA, HMK17wp}.
%
In our simple setting, the presence of only one traded security with constant dividends allows for an optimal continuous-time trading strategy that is absolutely continuous with respect to the Lebesgue measure ($dt$), even though Brownian noise enters the economy through the income streams.  Since consumption occurs on the $dt$-time scale, the absolute continuity property will allow transaction costs to be paid on the same $dt$-time scale in the real goods market.

In this section, we consider the continuous-time infinite time horizon Radner equilibrium.  The only traded security is an annuity $A$, which is in one-net supply and available to trade with the proportional transaction cost rate $\lambda\in[0,1)$.  The risk-free rate $r>0$ will be determined endogenously in equilibrium using the equilibrium annuity values.  We again focus on equilibria allowing for constant, positive interest rates, in which case, the annuity value will be the constant $A=1/r$.

Each of the two agents has an exogenous income stream given by $Y_i=(Y_{it})_t$, $i=1,2$, with dynamics 
$$
  dY_{it} = \mu_i dt + \sigma_i dB_{it}, \ \ \  Y_{i0}\in\R,
$$
where $\mu_i\in\R$, $\sigma_i>0$, and $B_1$ and $B_2$ are possibly correlated Brownian motions.  %Because the agents are unable to trade away the risk associated with their income streams, we make no assumptions on the relationship between $B_1$ and $B_2$.  
The agents are also endowed with an initial allocation of shares in the annuity $\theta_{i0}\in\R$ such that $\theta_{10}+\theta_{20} = 1$.

The flow of information in the economy is given by $\bF=(\sF_t)_{t\geq 0}$, where $\sF_t=\sigma(B_{1u},B_{2u}: 0\leq u\leq t)$.  All process are assumed to be adapted to $\bF$, and all agents share the same filtration.

\subsection{Individual Agent Problems}

We consider the single agent investment and consumption problem for agent $i$'s shadow market, $i\in\{1,2\}$. We focus on models where the value of the shadow annuity and shadow interest rate (to be determined endogenously in equilibrium) are constants $A_{i t}=A_i=1/r_i$ and $r_i>0$, respectively.

For a given investment strategy $\theta$, agent $i$'s shadow wealth is defined by $X_{i t} := \theta_{t} A_{i t}$.  For a measurable, adapted consumption process $c=(c_t)_t$ for which $\int_0^T |c_t| dt<\infty$ $\bP$-almost surely for all $T>0$,  the shadow wealth process associated with $c$ evolves like
$$
  dX^c_{it} = \left(X^c_{it}r_i - c_t + Y_{it}\right)dt, 
  \ \ \ X^c_{i0}=\theta_{i0}/r_i \in\R.
$$

As in the discrete-time case, we consider agents with exponential preferences over running consumption with risk aversion $\alpha_i>0$ and time-preference parameter $\beta_i>0$.

\begin{definition}
  Let $i\in\{1,2\}$.  A consumption process $c=(c_t)_t$ is called 
  \textit{admissible for agent $i$} if %$\int_{0}^\infty \E\left[e^{-\alpha_i c_{t}}\right] dt<\infty$  and 
  the transversality condition holds:
  $$
    \lim_{t\rightarrow\infty} \E\left[e^{-\beta_i t-\alpha_i r_i X^c_t-\alpha_i Y_{it}} \right] = 0.
  $$
  In this case, we write $c\in\sA_i$.
\end{definition}

For $i\in\{1,2\}$, agent $i$'s value function is given by
$$
  V_i(x,y) := \sup_{c\in\sA_i} \E \int_0^\infty U_i(t, c_t)dt, 
  \ \ \  x,y\in\R.
$$ 
We show in Theorem \ref{thm:cts_individual_agent} (below) that $V_i=J_i$, where
\begin{equation}\label{eqn:cts_J}
  J_i(x,y) = -\frac{1}{r_i}\exp\left(-\alpha_ir_ix - \alpha_iy 
  + 1 - \frac{\tilde\beta_i}{r_i}\right).
\end{equation}
We note that $V^\Delta_i(x,y)\Delta \longrightarrow J_i(x,y)$ as $\Delta\rightarrow 0$, where $V^\Delta_i$ is the discrete-time value function defined by \eqref{eqn:discrete_value_function}.

The following result establishes the individual agent optimal investment strategies.  The proof is omitted, as it does not vary substantially from the discrete-time case.
\begin{theorem}\label{thm:cts_individual_agent}
  For $i=1,2$, let 
  $\tilde\beta_i:=\beta_i+\alpha_i\mu_i-\frac{\alpha_i^2\sigma_i^2}{2}$.  
  The optimal consumption policy and wealth process for agent 
  $i$ %\in\{1,2\}$ 
  are given by
\begin{gather}
  \hat c_{it} = r_i\hat X_{it} + Y_{it} +\frac{\tilde\beta_i}{r_i\alpha_i} 
    -\frac{1}{\alpha_i}, \label{eqn:cts_consumption}\\
  \hat X_{it} = X^{\hat c_i}_t 
    = \frac{\theta_{i0}}{r_i} 
    + \frac{1}{\alpha_i}\left(1-\frac{\tilde\beta_i}{r_i}\right)t.
    \label{eqn:cts_wealth}
\end{gather}
Moreoever, the value function coincides with \eqref{eqn:cts_J}; that is, $V_i = J_i$.
\end{theorem}

\subsection{Equilibrium in Continuous-Time}\label{section:cts-eq}
In addition to establishing the existence of an equilibrium, we are interested in how the equilibrium interest rate depends on $\lambda$.

\begin{definition}
For the transaction cost parameter $\lambda\in[0,1)$, an \textit{equilibrium with transaction costs} is given by a collection of processes $(A_i, \hat c_i, \hat\theta_i)_{i=1,2}$ such that
\begin{enumerate}[(i)]
  \item For $i=1,2$, the optimal investment strategy 
    $\hat\theta_i$ is differentiable in time with derivative 
    $\hat\theta_{it}'$.
  \item Real and financial markets clear for all $t\geq 0$:
    $$
      \sum_{i=1}^2\hat c_{i t} = 1 +\sum_{i=1}^2 Y_{it} 
      - 2\lambda\left| \hat\theta_{1t}' \right|A_{t} 
      \ \ \ \text{ and } \ \ \ 
      \hat\theta_{1t} + \hat\theta_{2 t} = 1,
    $$
    where in the event of a trade, we define $A_t:=\frac{A_{it}}{1+\lambda}$ 
    if agent $i\in\{1,2\}$ purchases a positive number of annuity shares; i.e.,  
    $\hat\theta_{it}'>0$.
  \item For each agent $i=1,2$, the consumption and investment strategies, 
  $\hat c_i$ and $\hat\theta_i$ with $\hat\theta_{i0}=\theta_{i0}$, are 
  optimal with the annuity price $A_i$:
    $$
      V_i(\theta_{i0}A_{i0}, Y_{i0}) = \int_{0}^\infty \E
      \left[U_i(t, \hat c_{i t})\right]dt.
    $$
  \item \label{def:cts_closeness} 
  The shadow markets remain ``close enough'' to the underlying transaction 
  cost market in the following sense:  For all $t\geq 0$,
  $$
    \frac{A_{1t}}{A_{2t}}\in\left[\frac{1-\lambda}{1+\lambda},\frac{1+\lambda}{1-\lambda}\right].
  $$
  Moreover, if $\hat\theta_{1t}'>0$ then $A_{1t} = A_{2t}\cdot\frac{1+\lambda}{1-\lambda}$.  
  If $\hat\theta_{1t}'<0$ then $A_{1t} = A_{2t}\cdot\frac{1-\lambda}{1+\lambda}$. 
\end{enumerate}
\end{definition}
\begin{remark}
  When trade occurs, we are able to define $A_t := \frac{A_{it}}{1+\lambda}$, where agent $i\in\{1,2\}$ purchases a positive number of shares, $\hat\theta'_{it}>0$.  In this case, if we consider constant interest rate equilibria, then the equilibrium interest rate is uniquely determined by $r = 1/A = 1/A_t$.
\end{remark}

The following result establishes an equilibrium for the continuous-time model. The proof is omitted, as it mirrors the proof of Theorem \ref{thm:discrete_equilibrium}.
\begin{theorem}\label{thm:cts_eq}
  %Assume that $\mu_i>\frac{\alpha_i\sigma_i^2}{2}$, and 
  Let $\tilde\beta_i:=\beta_i+\alpha_i\mu_i-\frac{\alpha_i^2\sigma_i^2}{2}$, 
  and assume that $\tilde\beta_i$ is strictly positive for $i=1,2$. 
  For $\lambda\in[0,1)$, there exists an equilibrium with strictly positive 
  constant shadow interest rates $r_1, r_2$ and constant shadow annuity 
  values $A_1=1/r_1$, $A_2=1/r_2$.  
  The optimal consumption policies and wealth processes are 
  given by \eqref{eqn:cts_consumption} and \eqref{eqn:cts_wealth}, 
  respectively.

  \noindent{\bf Case 1:}  A no-trade equilibrium occurs if 
  \begin{equation}\label{cts:case1}
    \frac{\tilde\beta_2}{\tilde\beta_1} 
    \in\left[\frac{1-\lambda}{1+\lambda}, \frac{1+\lambda}{1-\lambda}\right].
  \end{equation}
  In this case, $r_1 = \tilde\beta_1$ and $r_2=\tilde\beta_2$. The 
  range of possible (non-shadow) interest rates that are consistent with 
  this equilibrium is $r=(r_{t})_{t\geq 0}$ where 
  $r_{t}\in[(1-\lambda)\max(\tilde\beta_1,\tilde\beta_2),
    (1+\lambda)\min(\tilde\beta_1,\tilde\beta_2)] \neq \emptyset$.

  \noindent{\bf Case 2:}  There exists 
  an equilibrium in which agent $1$ will purchase shares of the annuity in equilibrium at all times $t\geq 0$ (while agent $2$ sells shares) if
  \begin{equation}\label{cts:case2} % i=1, j=2
    \frac{\tilde\beta_2}{\tilde\beta_1} > \frac{1+\lambda}{1-\lambda}.
  \end{equation}
  In this case, the interest rate $r>0$ is determined by
  \begin{equation}\label{eqn:cts_rate}
    r = \frac{\tilde\beta_1/\alpha_1+\tilde\beta_2/\alpha_2}
    {\frac{1}{\alpha_1(1+\lambda)}+\frac{1}{\alpha_2(1-\lambda)}}.
  \end{equation}
  The shadow interest rates are given by $r = (1+\lambda)r_1=(1-\lambda)r_2$. 
  %Moreover, this equilibrium is unique in the class of constant, positive 
  %interest rate equilibria.
\end{theorem}

The agents behave similarly in a continuous-time equilibrium as in discrete time.  When the agents' income-adjusted time preference parameters $\tilde\beta_1$ and $\tilde\beta_2$ are sufficiently close, as in \eqref{cts:case1}, then the agents are not motivated to trade because of the relatively high transaction costs.  Their shadow interest rates reflect their individual frictionless view of the market and do not differ significantly enough to encourage trade.

Trading occurs in equilibrium only when the parameters $\tilde\beta_1$ and $\tilde\beta_2$ are sufficiently far apart to overcome the transaction costs.  In \eqref{cts:case2}, agent 1 values the annuity  more (with a lower shadow interest rate) than agent 2, which encourages her to acquire shares in the annuity.  The agents' strategies are very simple:  either buy or sell at the same rate for all times.

% Delta --> 0
Theorem~\ref{thm:discrete_to_cts} proves that the discrete-time interest rate passes to the continuous-time equilibrium rate as the time step $\Delta$ tends to zero.
\begin{theorem}\label{thm:discrete_to_cts}
  Suppose that
  $$
    \frac{\tilde\beta_2}{\tilde\beta_1}>\frac{1+\lambda}{1-\lambda}. 
  $$
  Let $r(\Delta)$ be the solution to 
  \eqref{eqn:discrete_r} corresponding to the time step $\Delta>0$, 
  and let $r(0)$ be the continuous-time equilibrium interest rate given by 
  \eqref{eqn:cts_rate}.  
  Then $r(\Delta)>0$ is the unique interest rate among constant interest rate 
  equilibria for sufficiently small $\Delta$, and 
  $r(\Delta) \longrightarrow r(0)$ as $\Delta\rightarrow 0$.
\end{theorem}

\begin{remark}
When $\frac{\tilde\beta_1}{\tilde\beta_2}>\frac{1+\lambda}{1-\lambda}$, we can conclude analogous results to Theorem~\ref{thm:cts_eq} Case 2 and Theorem~\ref{thm:discrete_to_cts} in which the roles of agent $1$ and $2$ are interchanged.  An analogous result to Corollary~\ref{cor:small_lambda} holds for $\tilde\beta_1>\tilde\beta_2>0$.
\end{remark}

% effects of small TC
\section{Effects of Transaction Costs}\label{section:effects}
In this section, we analyze the effects of transaction costs on the equilibrium interest rate and agent welfare.  Our method of solving for an equilibrium is the same for zero and non-zero transaction costs, which allows us to easily compare the endogenous interest rate as $\lambda$ varies.  The simplicity of the continuous-time limiting model lends itself to further study as the transaction costs tend to zero.   Even in this stylized model, the equilibrium interest rate is impacted by frictions in a non-trivial way and is not always monotonic.

Zero transaction costs and traded randomness from the stochastic income streams lead to complete markets and Pareto optimal equilibrium allocations.  In our model, there is no market for the risk associated with the agents' stochastic income streams.  Consequently, the equilibrium allocation is non-Pareto optimal, even when transaction costs are zero.  In Section~\ref{section:welfare} below, we use certainty equivalents to compare the welfare loss due to  transaction costs and unspanned income in equilibrium, similar to \cite{D16wp}.  Both types of incompleteness lead to a welfare loss.
%Welfare is lost in the zero transaction cost equilibrium compared to a complete market equilibrium, in which the stochastic income risk is spanned through trading. %, the zero transaction cost equilibrium allocation is non-Pareto optimal because the stochastic income risk is not traded.  
%Under some configurations of input parameters, it is possible to recover some welfare loss by imposing a strictly positive transaction cost, related to the \textit{theory of the $n$th (or second) best} as in \cite{H64AER, LL56RES}.  A similar phenomenon is described in \cite{D16wp}, where a strictly positive proportional transaction cost can sometimes recover lost welfare.

\subsection{Interest Rate Effects}\label{section:rate-effects}
For $i=1,2$, we recall that the risk- and income-adjusted time preference parameters are given by $\tilde\beta_i:=\beta_i+\alpha_i\mu_i-\frac{\alpha_i^2\sigma_i^2}{2}$.  When $\tilde\beta_1$ and $\tilde\beta_2$ differ and are strictly positive, then Theorem~\ref{thm:cts_eq} establishes the existence of a continuous-time equilibrium in which trade will occur for sufficiently small transaction costs.  We define a transaction cost threshold $\hat\lambda\in\R$ by 
$$
  \hat\lambda := \frac{\left(\sqrt{\alpha_2}-\sqrt{\alpha_1}\right)^2}{\alpha_2-\alpha_1}.
$$
Corollary \ref{cor:small_lambda} describes the behavior of the equilibrium interest rate in the case when trade occurs.  The proof is contained in Section~\ref{section:proofs}.

\begin{corollary}[The Effects of Small Transaction Costs]\label{cor:small_lambda}
Suppose that $\tilde\beta_2>\tilde\beta_1>0$.  For $\lambda\in [0,\frac{\tilde\beta_2-\tilde\beta_1}{\tilde\beta_1+\tilde\beta_2})$, the equilibrium interest rate exists, is unique among constant interest rate equilibria, and has the explicit form given by 
$$
  r = r(\lambda) 
  = \frac{\tilde\beta_1/\alpha_1+\tilde\beta_2/\alpha_2}
    {\frac{1}{\alpha_1(1+\lambda)}+\frac{1}{\alpha_2(1-\lambda)}}.$$
The transaction cost threshold $\hat\lambda$ is strictly positive when $\alpha_1<\alpha_2$.  In this case, $r$ is strictly increasing on $\left(0,\min\left(\frac{\tilde\beta_2-\tilde\beta_1}{\tilde\beta_1+\tilde\beta_2},\hat\lambda\right)\right)$ and strictly decreasing on $\left(\min\left(\frac{\tilde\beta_2-\tilde\beta_1}{\tilde\beta_1+\tilde\beta_2},\hat\lambda\right), \frac{\tilde\beta_2-\tilde\beta_1}{\tilde\beta_1+\tilde\beta_2}\right)$.  The equilibrium interest rate $r$ is strictly decreasing when $\alpha_1\geq\alpha_2$, in which case $\hat\lambda\leq 0$.
\end{corollary}

The impact of transaction costs on the equilibrium interest rate depends on the configuration of agent risk aversions $\alpha_i$ and risk- and income-adjusted time preference parameters $\tilde\beta_i$.  An agent with a large $\tilde\beta_i$ is a planner:  For small enough transaction costs, she will choose to forgo consumption now in order to invest in shares of the annuity to ensure a future consumption stream.  An agent with a small $\tilde\beta_i$ is a consumer:  For small enough transaction costs, she will sell annuity shares in order to consume today, even though she will miss out on the future dividend stream.

In addition to the adjusted time preference parameters, the agents' risk aversion levels determine the effect of transaction costs on the equilibrium interest rate and welfare levels. %XXX check this for welfare!
Table~\ref{table:regimes} below outlines the four possible agent configurations  that assist in defining two equilibrium regimes.

\begin{tablebox}

\begin{center}
  \begin{tabular}{ c | c | c |}\label{table:regimes}
      & small risk              & large risk \\
      & aversion,  $\alpha_i>0$ & aversion,  $\alpha_i>0$ \\ \hline
    small adjusted & \textbf{Aggressive Planner} & \textbf{Reserved Planner} \\
    time preferences, & invests in the annuity & invests in the annuity \\
    $\tilde\beta_i>0$ & prefers future consumption & prefers future consumption \\
      & most risk seeking & most risk averse \\ \hline
    large adjusted & \textbf{Aggressive Consumer} & \textbf{Reserved Consumer} \\
    time preferences, & sells the annuity & \text{sells the annuity} \\
    $\tilde\beta_i>0$ & prefers consumption today & prefers consumption today \\
      & most risk seeking & most risk averse \\ \hline
  \end{tabular}
\end{center}
\end{tablebox}

% consumption now versus later
In the following analysis, we consider the configuration of agent parameters from Corollary~\ref{cor:small_lambda} in which $\tilde\beta_2>\tilde\beta_1>0$ so that agent 1 buys while agent 2 sells for sufficiently small transaction costs.  By investing in shares of the annuity, agent 1 plans for the future by sacrificing her consumption goods today for the annuity's guaranteed dividend streams later.  Agent 2 seeks to consume now.  He sells shares in the annuity in order to earn immediate consumption units from the sale while forgoing the annuity's future dividend stream.  We consider two configurations of the agents' risk aversions in order to determine the impact of transaction costs on the interest rate and welfare. \\

\noindent\textbf{Case 1: $\alpha_1\geq\alpha_2$.} The configuration $\tilde\beta_2>\tilde\beta_1>0$ and $\alpha_1>\alpha_2>0$ corresponds to a reserved planner (agent 1) and aggressive consumer (agent 2) economy.  Agent 1 is more reserved in terms of her risk aversion and time preferences, while agent 2's more cavalier attitude prods him to take on additional risk in order to consume more now.  As frictions are introduced into the economy, the agents' contrasting risk and time preferences translate into a lower premium placed on the equilibrium interest rate.\\

\noindent\textbf{Case 2: $\alpha_1<\alpha_2$.} The configuration $\tilde\beta_2>\tilde\beta_1>0$ and $\alpha_2>\alpha_1>0$ corresponds to an aggressive planner (agent 1) and reserved consumer (agent 2) economy.  The equilibrium interest rate reflects the market's compensation for future consumption.  A small level of friction (up to the level $\hat\lambda$) benefits the aggressive planner as she is compensated for her desire to invest in the annuity to ensure her future consumption.  Her modest appetite for risk is reflected by her risk aversion parameter $\alpha_1$, which is dominated by $\alpha_2$ yet remains sufficiently small so that the agents are willing to trade with $\tilde\beta_1<\tilde\beta_2$.

When transaction costs become sufficiently large in that they go above the threshold $\hat\lambda$, then the aggressive planner no longer receives an additional interest rate premium for her future planning.  In this case, the interest rate decreases as transaction costs rise above $\hat\lambda$, and the reserved consumer begins to benefit from selling annuity shares in a lower rate environment.  In conclusion, for strictly positive but small transaction costs, a preference for future consumption and moderate amount of risk can be beneficial if the other agent is more risk averse and prefers immediate consumption.  For $\lambda>\hat\lambda$, the interest rate premium declines. % timid consumer

%The interest rate increase in Case 2 is related to the \textit{theory of the second best}, as described in \cite{}

Figure \ref{fig:rates} plots the equilibrium interest rate as a function of transaction costs $\lambda$ in the range $[0,\frac{\tilde\beta_2-\tilde\beta_1}{\tilde\beta_1+\tilde\beta_2})$ for three different input parameterizations and both cases described above.  We assume that the agents in Figure \ref{fig:rates} have identical income volatility $\sigma_i=.01$ and time preference parameters $\beta_i=.02$, $i=1,2$; see \cite{CGM05RFS}.  Their risk aversions $\alpha_i$ vary between $2$ and $8$, while their income drift $\mu_i$ is between $0.02$ and $0.1$, $i=1,2$; see, \cite{AP10JAE}, \cite{CGM05RFS} and \cite{S86RES}.  The given parameter specifications allow for each agent to have dominant risk aversion while still ensuring that the income-adjusted time preference parameters are ordered by $\tilde\beta_2>\tilde\beta_1>0$.  Though the input parameters are in some sense reasonable and the expression for the interest rate is explicit, the effects of transaction costs on the equilibrium interest rate vary significantly amongst the parameterizations, as illustrated in Figure \ref{fig:rates}.

\begin{figure}
  \center\includegraphics[width=\linewidth]{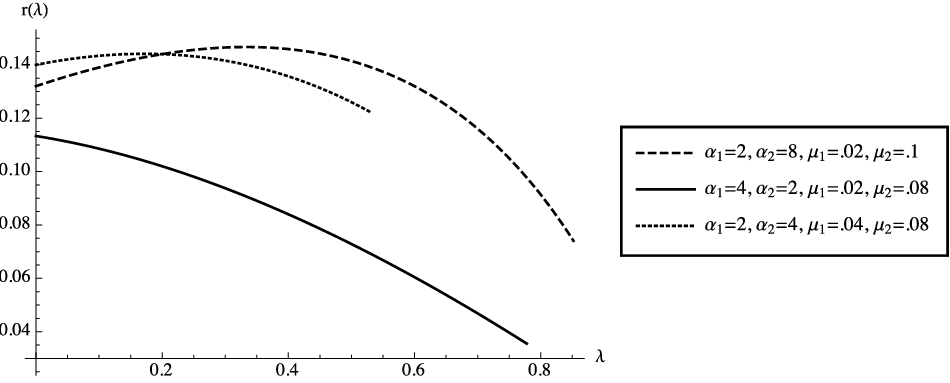}
  \caption{The equilibrium interest rate is plotted as a function of the transaction costs for $\lambda$ in the Case 2 trading range $[0,\frac{\tilde\beta_2-\tilde\beta_1}{\tilde\beta_1+\tilde\beta_2})$.  The solid line plot represents Case 1 with a reserved planner and aggressive consumer economy.  The dashed line plots represent Case 2 with an aggressive planner and reserved consumer economy.}
  \label{fig:rates}
\end{figure}

\subsection{Complete Market Equilibrium}\label{section:complete}

Our model's incompleteness stems from both the unspanned income streams and transaction costs.  When transaction costs are zero, it is possible to complete the market by introducing additional financial securities. In order to measure the welfare loss due to incompleteness, we compare our market to one with the same aggregate demand and agent preferences, where the risks from trading are spanned in a dynamically complete way by introducing risky financial assets.  Since traded financial securities span $B_1$ and $B_2$, we denote the correlation between $B_1$ and $B_2$ by $\rho\in[-1,1]$ so that $d\left<B_1,B_2\right>_t = \rho dt$.

% repr agent \sU^\gamma
For $i=1,2$, we recall that $U_i(t,x) =  -e^{-\beta_i t-\alpha_i x}$.  We consider a representative agent with weight $\gamma>0$, whose indirect utility for a given consumption stream $c$ is given by the sup-convolution,
$$
  \sU^\gamma(c) := \sup_{\stackrel{c_1\in\sA_1, c_2\in\sA_2:}{\forall t,\ c_{1t}+c_{2t}\leq c_t}} \E\left[\int_0^\infty \left(U_1(t,c_{1t})+\gamma U_2(t,c_{2t})\right)dt\right].
$$
We define the representative risk aversion parameter $\alpha_r>0$ and representative time preference parameter $\beta_r$ by
$$
  \alpha_r := \frac{1}{1/\alpha_1 + 1/\alpha_2}\ \ \ \text{ and } \ \ \ 
  \beta_r := \frac{\beta_1/\alpha_1 +\beta_2/\alpha_2}{1/\alpha_1 + 1/\alpha_2}.
$$
A corresponding representative utility function is given by $U^\gamma(t,x) = e^{-\beta_r t-\alpha_r x}$. For the  aggregate demand, $c=Y_1+Y_2+1$, we have that
\begin{align*}
  \sU^\gamma(Y_1+Y_2+1) 
    &= -\E\left[\int_0^\infty U^\gamma(t,Y_{1t}+Y_{2t}+1)dt\right]\\
    &= -\E\left[\int_0^\infty e^{-\beta_r t -\alpha_r(Y_{1t}+Y_{2t}+1)}\cdot \frac{\alpha_1+\alpha_2}{\alpha_2}\left(\frac{\alpha_1}{\alpha_2\gamma}\right)^{-\frac{\alpha_1}{\alpha_1+\alpha_2}}dt\right] \\
  &= -\frac{\frac{\alpha_1+\alpha_2}{\alpha_2}\left(\frac{\alpha_1}{\alpha_2\gamma}\right)^{-\frac{\alpha_1}{\alpha_1+\alpha_2}} e^{-\alpha_r(y_1+y_2+1)}}{\beta_r+\alpha_r(\mu_1+\mu_2)-\frac{\alpha_r^2}{2}(\sigma_1^2+\sigma_2^2+2\rho\sigma_1\sigma_2)},
\end{align*} % XXX need to figure out for which configuration of input parameters this integral is finite and how it differs from incomplete case
and the corresponding optimal individual consumptions $c_1, c_2$ given in the sup-convolution are
\begin{gather*}% resulting optimal agent consumption streams
  c^\gamma_{1t} 
    = \frac{\alpha_2}{\alpha_1+\alpha_2}(Y_{1t}+Y_{2t}+1)
    -\frac{\beta_1-\beta_2}{\alpha_1+\alpha_2}t
    +\frac{1}{\alpha_1+\alpha_2}\log\left(\frac{\alpha_1}{\gamma\alpha_2}\right),\\
  c^\gamma_{2t} 
    = \frac{\alpha_1}{\alpha_1+\alpha_2}(Y_{1t}+Y_{2t}+1)
    +\frac{\beta_1-\beta_2}{\alpha_1+\alpha_2}t
    -\frac{1}{\alpha_1+\alpha_2}\log\left(\frac{\alpha_1}{\gamma\alpha_2}\right).
\end{gather*}

The equilibrium state price density $\xi=(\xi_t)_t$ is described by the first-order condition for the aggregate consumption by $\xi_t = \frac{U^\gamma_c(t,Y_{1t}+Y_{2t}+1)}{U^\gamma_c(0,y_1+y_2+1)}$, where $U^\gamma_c$ denotes the derivative in the consumption variable of $U^\gamma$.  The dynamics of the state price density are given by
$$
  d\xi_t = -\xi_t(\rcmp_t dt + \nu_{1t}dB_{1t}+ \nu_{2t}dB_{2t}), \ \ \xi_0=1,
$$
where $\rcmp$ is the equilibrium interest rate, and $\nu_1$ and $\nu_2$ are the market prices of risk corresponding to the Brownian motions $B_1$ and $B_2$, respectively.  The equilibrium interest rate in the complete market is computed using the first-order condition and $\xi$'s dynamics.  It is given by
$$% using the equilibrium state price density, find the resulting eq interest rate
  \rcmp = \beta_r + \alpha_r(\mu_1+\mu_2) -\frac{\alpha_r^2}{2}(\sigma_1^2+\sigma_2^2+2\rho\sigma_1\sigma_2).
$$
We note that $\xi$, $\rcmp$, $\nu_1$, and $\nu_2$ are independent of the weighting superscript $\gamma$.  Since the agents' preferences are described by exponential utility functions, these terms do not depend in equilibrium on $\gamma$.

We are interested in the welfare level of a complete market economy and will use the sum of the agents' certainty equivalents as a proxy for welfare.
\begin{definition} For $i=1,2$ and representative agent weight $\gamma>0$, a value $CE_i^{comp}$ is called the \textit{certainty equivalent} for agent $i$ if 
$$% XXX fix this defn up -- need allow for dependence on \gamma...
  \int_0^\infty U_i\left(t,CE_i^{comp}\right) dt 
  = -\int_0^\infty U_i\left(t, c^\gamma_{it}\right) dt,
$$
where $c^\gamma_{i}$ is agent $i$'s optimal consumption stream.  We write $CE_i^{comp}(\gamma)$ to emphasize the dependence of the certainty equivalent on $\gamma$.
\end{definition}

The certainty equivalent represents a constant consumption stream level that an agent is willing to exchange for her optimal (stochastic) consumption stream. For $i\in\{1,2\}$ and $\gamma>0$, the certainty equivalent is given by
$$
  CE_i^{comp}(\gamma) = \frac{1}{\alpha_i} \log\left(\frac{\rcmp}{\beta_i}\right)+c^\gamma_{i0}.
$$
The sum of the certainty equivalents can be used as a welfare measure in the economy.  In the complete market equilibrium, for initial wealth stream values $y_1, y_2\in\R$ and representative agent weight $\gamma>0$, we have
$$
  CE_1^{comp}(\gamma)+CE_2^{comp}(\gamma) = 1+y_1+y_2+\frac{1}{\alpha_1} \log\left(\frac{\rcmp}{\beta_1}\right)+\frac{1}{\alpha_2} \log\left(\frac{\rcmp}{\beta_2}\right).
$$
We note that the sum of the certainty equivalents does not depend on the weight $\gamma>0$.

\subsection{Welfare Loss from Transaction Costs}\label{section:welfare}
In this section, we compare the agents' welfare loss due to transaction costs.  Similar to the approach in \cite{D16wp}, we use the sum of the agents' certainty equivalents as our measure of welfare.  We find that the introduction of transaction costs causes a strict loss of welfare.

\label{note:davilla}Our results contrast \cite{D16wp}, which performs a related analysis in a one-period continuum-of-agents model with heterogenous beliefs on a risky asset, an exogenous riskless asset, and no income.  \cite{D16wp} finds that under some parameter specifications, a strictly positive transaction cost can provide a welfare gain.  Our results from Section~\ref{section:rate-effects} show that it may be possible for the interest rate to increase for a strictly positive level of transaction costs, yet the welfare itself cannot be recouped.

\begin{definition} For $i=1,2$, $\lambda\in[0,1)$ and $(x,y)\in\R^2$, a value $CE_i$ is called the \textit{certainty equivalent} for agent $i$ at transaction cost level $\lambda$, initial wealth $x$, and initial income level $y$ if 
$$
  \int_0^\infty U_i\left(t,CE_i\right) dt = V_i(x,y).
$$
We write $CE_i(\lambda)$ to emphasize the dependence of the certainty equivalent on $\lambda$.
\end{definition}

The certainty equivalent represents a constant consumption stream level that an agent is willing to exchange for her optimal (stochastic) consumption stream.  For a given shadow interest rate $r_i>0$, the certainty equivalent can be expressed as
$$
  CE_i = \frac{1}{\alpha_i}\log\left(\frac{r_i}{\beta_i}\right) + \hat c_{i0},
$$
where $\hat c_{i0}$ is the optimal consumption level at time $0$ and is given in \eqref{eqn:cts_consumption}.

In equilibrium, we are interested in the sum of our agents' certainty equivalents as a proxy for the welfare of the economy.  For $\tilde\beta_2>\tilde\beta_1>0$,
\begin{gather*}
  CE_1(\lambda)+CE_2(\lambda) =
  \begin{cases}
     1+y_1+y_2+\frac{1}{\alpha_1}\left(\log\left(\frac{r(\lambda)}{\beta_1(1+\lambda)}\right)
     +\frac{\tilde\beta_1(1+\lambda)}{r(\lambda)}-1\right) & \\
     \ \ \ \ \ \ \ \ \ \ \ \ \ \ \ \ +\frac{1}{\alpha_2}\left(\log\left(\frac{r(\lambda)}{\beta_2(1-\lambda)}\right)+\frac{\tilde\beta_2(1-\lambda)}{r(\lambda)}-1\right),  & \text{if } \lambda<\frac{\tilde\beta_2-\tilde\beta_1}{\tilde\beta_1+\tilde\beta_2},\\
     1+y_1+y_2 + \frac{1}{\alpha_1}\log\left(\frac{\tilde\beta_1}{\beta_1}\right)+\frac{1}{\alpha_2}\log\left(\frac{\tilde\beta_2}{\beta_2}\right), & \text{if } \lambda\geq\frac{\tilde\beta_2-\tilde\beta_1}{\tilde\beta_1+\tilde\beta_2}.
  \end{cases}
\end{gather*}
The following result states that the economy's welfare is decreasing in incompleteness due to transaction costs and unspanned income.  The proof of Proposition~\ref{prop:ce} is given in Section~\ref{section:proofs}.
\begin{proposition}\label{prop:ce}
  Suppose that $\tilde\beta_2>\tilde\beta_1>0$.  $CE_1(\lambda)+CE_2(\lambda)$ is strictly decreasing on $[0,\frac{\tilde\beta_2-\tilde\beta_1}{\tilde\beta_1+\tilde\beta_2})$ and constant on $[\frac{\tilde\beta_2-\tilde\beta_1}{\tilde\beta_1+\tilde\beta_2},1]$.  Moreoever, complete market welfare levels dominate incomplete market welfare levels in that for all transaction costs $\lambda\in[0,1]$ and weights $\gamma>0$, we have $CE_1(\lambda)+CE_2(\lambda)<CE_1^{comp}(\gamma)+CE_2^{comp}(\gamma)$.
\end{proposition}

Figure \ref{fig:ce} plots the economy's welfare change due to incompleteness as a function of transaction costs $\lambda$ for the three input parameterizations that were used in Figure~\ref{fig:rates}.  We assume that the income stream correlation is zero in that $\left<B_1,B_2\right>_t=0$ as in \cite{CGM05RFS}. The given parameter specifications allow for each agent to have dominant risk aversion while still ensuring that the income-adjusted time preference parameters are ordered by $\tilde\beta_2>\tilde\beta_1>0$.  Regardless of the input parameterizations, the sum of the certainty equivalents, $CE_1+CE_2$, is decreasing in the transaction costs.  When no trading occurs, the certainty equivalent sum is constant in transaction costs.  The economy's welfare always decreases when moving from a complete to incomplete market.  Though the interest rate is possibly non-monotone in the transaction cost level as in Figure~\ref{fig:rates}, the welfare in the economy only decreases.

\begin{figure}
  \center\includegraphics[width=\linewidth]{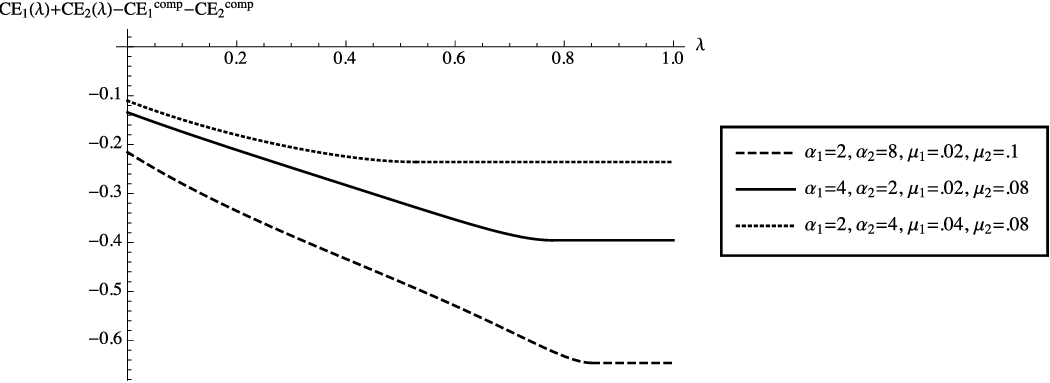}
  \caption{The welfare change due to incompleteness is plotted as a function of the transaction costs. The solid line plot represents Case 1 with a reserved planner and aggressive consumer economy.  The dashed line plots represent Case 2 with an aggressive planner and reserved consumer economy.  The economy's welfare is decreasing in both cases.}
  \label{fig:ce}
\end{figure}

% Bank account
\section{The Bank Account as the Traded Security}\label{section:bank}
Transaction costs in our model prevent us from trading freely between the annuity and a bank account.  Using an annuity as our traded security allows for constant shadow interest rates and trading strategies that are the same at every time point:  either the agents buy, sell, or trade nothing.  The simple structure of transaction cost equilibria with a traded annuity is not possible when the bank account is traded instead.

In this section, we consider a discrete-time equilibrium with transaction costs when the bank account is the traded security.  Theorem~\ref{thm:bank} proves that the traded bank account model prevents a constant-interest rate transaction cost equilibrium.

In contrast to the annuity, the bank account is a financial asset in zero-net supply.  For $i=1,2$, the shadow bank account $B_i$ has the associated interest rate process $r_i=(r_{it_n})_{n\geq 0}$ and is given by $B_{i0}=1$ and 
$$
  B_{it_n} = (1+r_{i0}\Delta)\cdot \ldots\cdot (1+r_{it_{n-1}}\Delta), 
  \ \ \ n\geq 1.
$$
We focus on equilibria yielding constant shadow interest rates $r_{it_n} = r_i$, as in the traded annuity case.

For a given investment strategy $\theta$, agent $i$'s shadow wealth is given by $X_{it_n}:=\theta_{t_n}B_{it_n}$. Since the bank account is in zero-net supply, the self-financing condition in \eqref{eqn:discrete_sf} will be replaced by
$$
  (\theta_{t_{n+1}}-\theta_{t_{n}})B_{i t_{n+1}} 
  = \left(Y_{i t_n}-c_{t_n}+\theta_{t_n}\right)\Delta, \ \ \ n\geq 0.
$$
Thus, for a given consumption and investment strategy $(c,\theta)$, %with $\theta_0 = \theta_{i0}$, 
the shadow wealth evolves like
$$
  X^c_{it_{n+1}} - X^c_{it_n} = \left(X^c_{it_n}r_i+Y_{it_n}-c_{t_n}\right)\Delta, \ \ \ X_{i0}=\theta_{0} B_{i0} = \theta_{0}.
$$
The definitions of admissibility and the value function are unchanged from Definition~\ref{def:discrete_admissibility} and \eqref{eqn:discrete_value_function}.  As such, Theorem~\ref{thm:discrete_individual} holds for the frictionless shadow market with a bank account carrying a constant interest rate.

\begin{definition}\label{def:bank_equilibrium}
For the transaction cost parameter $\lambda\in[0,1)$, a \textit{transaction cost equilibrium with a bank account} is given by a collection of processes $(r_i, \hat c_i, \hat\theta_i)_{i=1,2}$ such that
\begin{enumerate}[(i)]% definition not edited below:
  \item \label{def:bank_clearing}
  Real and financial markets clear for each $n\geq 0$:
    $$
      \sum_{i=1}^2\hat c_{i t_n}\Delta = \sum_{i=1}^2 Y_{i t_n}\Delta 
      - 2\lambda\left|\hat\theta_{1 t_{n+1}}-\hat\theta_{1 t_n}\right|B_{t_{n+1}} 
      \ \ \ \text{ and } \ \ \ 
      \hat\theta_{1 t_n} + \hat\theta_{2 t_n} = 0,
    $$
    where in the event of a trade, we define 
    $B_{t_{n+1}}:=\frac{B_{i t_{n+1}}}{1+\lambda}$ 
    if agent $i\in\{1,2\}$ purchases a positive number of annuity shares; 
    that is, $\theta_{i t_{n+1}}-\theta_{i t_n}>0$.
  \item For each agent $i=1,2$, the consumption and investment strategies, $\hat c_i$ and $\hat\theta_i$ with $\hat\theta_{i0}=\theta_{i0}$, are optimal with the shadow bank account value $B_i$:
    $$ 
      V^\Delta_i(\theta_{i0}) = -\sum_{n=0}^\infty \E
      \left[e^{-\beta_i t_n}e^{-\alpha_i \hat c_{i t_n}}\right].
    $$
  \item \label{def:bank_closeness} 
  The shadow markets remain ``close enough'' to the underlying transaction 
  cost market in the following sense:  For each $n\geq 1$,
  $$
    \frac{B_{1 t_n}}{B_{2 t_n}}
    \in\left[\frac{1-\lambda}{1+\lambda},\frac{1+\lambda}{1-\lambda}\right].
  $$
  Moreover, if $\hat\theta_{1 t_n}-\hat\theta_{1t_{n-1}}>0$ then 
  $B_{1 t_n} = B_{2 t_n}\cdot\frac{1+\lambda}{1-\lambda}$.  
  If $\hat\theta_{1 t_n}-\hat\theta_{1 t_{n-1}}<0$ then $B_{1 t_n} = B_{2 t_n}\cdot\frac{1-\lambda}{1+\lambda}$. 
\end{enumerate}
\end{definition}

Theorem~\ref{thm:bank} shows that aside from a stylized special case, any transaction cost equilibrium with a bank account must have non-constant interest rates.  The proof is presented in Section~\ref{section:proofs}.
\begin{theorem}\label{thm:bank}
  Let $\tilde\beta_i:=\beta_i+\alpha_i\mu_i-\frac{\alpha_i^2\sigma_i^2}{2}$, 
  and suppose that $\tilde\beta_i$ and $\lambda$ are strictly positive 
  for $i=1,2$.  
  Suppose that $(r_i,\hat c_i,\hat\theta_i)_{i=1,2}$ is a transaction cost 
  equilibrium with a bank account.  
  If $r_1,r_2$ are strictly positive constants, then the following must hold:
  \begin{enumerate}[(1)]
    \item The agents' parameters satisfy $\tilde\beta_1=\tilde\beta_2$.
    \item No trading occurs in equilibrium: 
      $\hat\theta_{it_n}-\hat\theta_{it_{n-1}} = 0$ for $i=1,2$ and $n\geq 1$.
    \item The shadow rates are identical and satisfy
      $$
        r_1 = r_2 
        = \frac{1}{\Delta}\left(e^{\tilde\beta_1\Delta}-1\right).
      $$
  \end{enumerate}
\end{theorem}

\label{comment:bank2}Though it is possible to consider stochastic interest rates in a transaction cost equilibrium with a bank account using a system of variational inequalities, it is not clear if such an equilibrium exists.  The simple mathematical structure of the annuity cannot be obtained by studying a traded bank account in its place.  The annuity provides a constant dividend stream at all possible consumption times, which allows agents to receive constant future dividends without trading or incurring transaction costs.

% proofs
\section{Proofs}\label{section:proofs}
We begin by proving Theorem \ref{thm:discrete_individual} in the discrete-time case.
\begin{proof}
  We check that $\hat c_{i}$ is admissible, by noting that 
  \begin{align*}
    \E
    &\left[\exp\left(-\beta_i t_n-\alpha_ir_i X^{\hat c_i}_{it_n}
      -\alpha_i Y_{it_n}\right)\right]\\
    &= \E\left[\exp\left(-\alpha_ir_iX_{i0} -n \log(1+r_i\Delta)
      -\frac{\alpha_i^2\sigma_i^2}{2}n\Delta 
      - \alpha_i\sum_{k=1}^n\sqrt{\Delta}\sigma_i Z_{i t_k}\right)\right] \\
    &= \left(1+r_i\Delta\right)^{-n}\exp\left(-\alpha_ir_iX_{i0}\right)
      \longrightarrow 0 
      \ \ \ \text{ as $n\rightarrow\infty$.}
  \end{align*}
  %Given the consumption strategy, $\hat c_i$, $\hat X_i$ is the wealth 
  %process corresponding to this consumption process, 
  %$\hat X_i = X^{\hat c_i}_i$.
  We have that $c\mapsto -e^{-\alpha_i c} + e^{-\beta_i\Delta} 
  \E\left[J\left(x(1+r_i\Delta)+y-c, 
  y+\mu_i\Delta + \sigma_i\sqrt{\Delta}Z\right)\right]$ 
  is maximized for $\hat c = \hat c(x,y) = r_i x +y 
  +\frac{\tilde\beta_i}{\alpha_i r_i} 
  - \frac{1}{\alpha_ir_i\Delta}\log(1+r_i\Delta)$, where $Z$ denotes a 
  standard normal random variable. Thus,
  $$
    \left\{-\sum_{k=0}^{n-1}e^{-\alpha_i c_{t_n}} 
    + e^{-\beta_i t_n}J^\Delta_i\left(X^c_{i t_n}, Y_{i t_n}\right)
    \right\}_{n\geq 0}
  $$
  is a supermartingale for all $c\in\sA^\Delta_i$ and is a martingale 
  for $c = \hat c_i\in\sA^\Delta_i$.
  
  Therefore, for $\hat c_i$,
  \begin{align*}
    J(x,y) &= -\E\left[\sum_{k=0}^n e^{-\alpha_i\hat c_{i t_k}}\right]
        + e^{-\beta_i t_{n+1}}\E\left[ J^\Delta_i
        \left(X^{\hat c_i}_{i t_{n+1}}, Y_{i t_{n+1}}\right)\right] \\
      &= -\E\left[\sum_{k=0}^\infty e^{-\alpha_i \hat c_{i t_k}}\right] 
        \ \ \ \ \text{ by the transversality condition},
  \end{align*}
  which implies that $J^\Delta_i \leq V^\Delta_i$.  Similarly, for any 
  $c\in\sA^\Delta_i$, 
  \begin{align*}
    J(x,y) &\geq -\E\left[\sum_{k=0}^n e^{-\alpha_i c_{t_k}}\right]
        + e^{-\beta_i t_{n+1}}\E\left[ J^\Delta_i
        \left(X^{c}_{i t_{n+1}}, Y_{i t_{n+1}}\right)\right] \\
      &= -\E\left[\sum_{k=0}^\infty e^{-\alpha_i c_{t_k}}\right] 
        \ \ \ \ \text{ by the transversality condition}.
  \end{align*}
  Thus, $J^\Delta_i = V^\Delta_i$, and $\hat c_i\in\sA^\Delta_i$ is the 
  optimal consumption policy.  For initial wealth 
  $x = \theta_{i0}A_i=\theta_{i0}/r_i$, 
  the optimal wealth policy corresponding to $\hat c_i$ is $\hat X_i = X^{\hat c_i}_i$, and 
  $$
    \hat X_{i t_n} = \frac{\theta_{i0}}{r_i} + \frac{t_n}{\alpha_i r_i}
    \left(\frac{1}{\Delta}\log\left(1+r_i\Delta\right)
    -\tilde\beta_i\right).
  $$
\end{proof}

We now move towards the proof of Theorem \ref{thm:discrete_equilibrium}.  The self-financing condition \eqref{eqn:discrete_sf} with the optimal policies \eqref{thm:discrete_individual} and \eqref{eqn:discrete_opt_consumption} imply
\begin{equation}\label{eqn:discrete_derive_F}
  \left(Y_{i t_n}-\hat c_{i t_n} + \hat\theta_{i t_n}\right)\Delta
  = \left(\hat\theta_{it_n}-\hat\theta_{i t_{n-1}}\right) A_{i t_n}
  = \frac{1}{\alpha_ir_i}\left(\log(1+r_i\Delta)-\tilde\beta_i\Delta\right).
\end{equation}
For $i=1,2$, we define% a function $F_i:(0,\infty)\rightarrow\R$ by 
\begin{equation}\label{def:discrete_F}
  F_i(r):= \frac{1}{\alpha_i r}\left(\log\left(1+r\Delta\right)-\tilde\beta_i\Delta\right), \ \ \ r>0.
\end{equation}
Using Definition \ref{def:discrete_equilibrium} part \eqref{def:discrete_clearing}, we seek solutions $r_1, r_2>0$ such that
$$
  F_1(r_1)+F_2(r_2) = \lambda\left(\frac{\left|F_1(r_1)\right|}{1+\lambda}+\frac{\left|F_2(r_2)\right|}{1-\lambda}\right)1_{\{F_1(r_1)\geq 0\}}+\lambda\left(\frac{\left|F_1(r_1)\right|}{1-\lambda}+\frac{\left|F_2(r_2)\right|}{1+\lambda}\right)1_{\{F_1(r_1)< 0\}}.
$$
By rewriting this equation and including Condition \eqref{def:discrete_closeness} from Definition \ref{def:discrete_equilibrium}, we seek $r_1, r_2>0$ such that
\begin{equation}\label{eqn:clearing_cond_const}
  F_{2}(r_2) = 
  \begin{cases}
    -\frac{1-\lambda}{1+\lambda}F_{1}(r_1), & \text{if } F_{1}(r_1)\geq 0,\\
    -\frac{1+\lambda}{1-\lambda}F_{1}(r_1), & \text{if } F_{1}(r_1)\leq 0.
  \end{cases}
\end{equation}
and
\begin{equation}\label{eqn:shadow_constraint_const}
  \frac{r_2}{r_1} = 
  \begin{cases}
    \frac{1+\lambda}{1-\lambda}, & \text{if } F_{1}(r_1)> 0,\\
    \frac{1-\lambda}{1+\lambda}, & \text{if } F_{1}(r_1)< 0,\\
    \in\left[\frac{1-\lambda}{1+\lambda},\frac{1+\lambda}{1-\lambda}\right], & \text{if } F_{1}(r_1)= 0,\\
  \end{cases}
\end{equation}
% mu>alpha*sigma^2/2 only required for case 2, although case 1's r's would be <=0 if either beta~<0.  So check the F eqns to see if beta~>0 needed.
\begin{proposition}\label{prop:discrete_F}
Let $\tilde\beta_i:= \beta_i+\alpha_i\mu_i - \alpha_i^2\sigma_i^2/2$, 
and suppose that $\tilde\beta_i$ is strictly positive for $i=1,2$. 
There exists a unique strictly positive solution pair $r_1,r_2$ to \eqref{eqn:clearing_cond_const} and \eqref{eqn:shadow_constraint_const}.

\noindent{\bf Case 1:}  If $\frac{e^{\tilde\beta_1\Delta}-1}{e^{\tilde\beta_2\Delta}-1}\in\left[\frac{1-\lambda}{1+\lambda},\frac{1+\lambda}{1-\lambda}\right]$, then
\begin{equation}\label{eqn:discrete_rs}
  r_1 = \frac{e^{\tilde\beta_1\Delta}-1}{\Delta}
  \ \ \ \text{ and } \ \ \ 
  r_2 = \frac{e^{\tilde\beta_2\Delta}-1}{\Delta}.
\end{equation}

\noindent{\bf Case 2:}  If % i=1, j=2
we have $\frac{e^{\tilde\beta_2\Delta}-1}{e^{\tilde\beta_1\Delta}-1}>\frac{1+\lambda}{1-\lambda}$, then the unique positive solutions satisfy
$$
  r_1\in\left(\frac{e^{\tilde\beta_1\Delta}-1}{\Delta}, \frac{1-\lambda}{1+\lambda}\left(\frac{e^{\tilde\beta_2\Delta}-1}{\Delta}\right)\right)
  \ \ \ \text{ and } \ \ \ 
  r_2 = \frac{1+\lambda}{1-\lambda} r_1.
$$
\end{proposition}
\begin{proof}
We show the existence of the unique solution pair by examining both cases.  Suppose that 
$$
  \frac{e^{\tilde\beta_1\Delta}-1}{e^{\tilde\beta_2\Delta}-1}
  \in\left[\frac{1-\lambda}{1+\lambda},\frac{1+\lambda}{1-\lambda}\right].
$$
  Then $r_1$ and $r_2$ as in \eqref{eqn:discrete_rs} is the unique solution 
  to \eqref{eqn:clearing_cond_const} and \eqref{eqn:shadow_constraint_const} 
  such that $F_1(r_1) = F_2(r_2) = 0$.

To show uniqueness, we proceed by contradiction.  Assume for the sake of contradiction that there exist strictly positive solutions $r_1, r_2$ such that  
  $F_1(r_1)>0$.  We have that $F_1(r_1)>0$ if an only if 
  $F_2(r_2)<0$, $r_1\Delta>e^{\tilde\beta_1\Delta}-1$, and 
  $r_2\Delta<e^{\tilde\beta_2\Delta}-1$.  Then by 
  \eqref{eqn:shadow_constraint_const},
  $$
    \frac{e^{\tilde\beta_1\Delta}-1}{e^{\tilde\beta_2\Delta}-1}
    <\frac{r_1}{r_2}
    = \frac{1-\lambda}{1+\lambda} 
    \leq \frac{e^{\tilde\beta_1\Delta}-1}{e^{\tilde\beta_2\Delta}-1},
  $$
  which is a contradiction.  Here, we have used that 
  $\tilde\beta_1, \tilde\beta_2$ are strictly positive to ensure that 
  $e^{\tilde\beta_i\Delta}-1>0$.  The same argument applies to rule out 
  the case when $F_1(r_1)<0$ and $F_2(r_2)>0$. 
  Therefore, we must have that 
  $F_1(r_1) = F_2(r_2) = 0$, in which case 
  $r_1 = \frac{e^{\tilde\beta_1\Delta}-1}{\Delta}$ and 
  $r_2 = \frac{e^{\tilde\beta_2\Delta}-1}{\Delta}$.

  % case 2 discrete case
  We now consider the existence of a solution in Case 2.  For 
  $F_1(r_1)>0$, \eqref{eqn:clearing_cond_const} and 
  \eqref{eqn:shadow_constraint_const} reduce to solving for $r_1>0$ 
  such that 
  \begin{equation}\label{eqn:discrete_F_rewritten}
    \left(1+r_1\Delta\cdot\frac{1+\lambda}{1-\lambda}\right)^{1/\alpha_2}
    \left(1+r_1\Delta\right)^{1/\alpha_1}
    = \exp\left(\left(\frac{\tilde\beta_1}{\alpha_1}
    +\frac{\tilde\beta_2}{\alpha_2}\right)\Delta\right),
  \end{equation}
  while $r_2 = \frac{1+\lambda}{1-\lambda}\cdot r_1$.  The assumption that 
  $\tilde\beta_1,\tilde\beta_2$ are strictly positive ensures that the right 
  hand side of \eqref{eqn:discrete_F_rewritten} is strictly bigger than $1$. 
  We note that %the function 
  $x\mapsto \left(1+x\cdot\frac{1+\lambda}{1-\lambda}\right)^{1/\alpha_2}
  \left(1+x\right)^{1/\alpha_1}$ strictly increases from $1$ to $\infty$ 
  for $x\in[0,\infty)$.  Thus, there exists a unique solution $r_1>0$ to 
  \eqref{eqn:discrete_F_rewritten}.  
  Moreover, $\frac{e^{\tilde\beta_2\Delta}-1}{e^{\tilde\beta_1\Delta}-1}
  >\frac{1+\lambda}{1-\lambda}$ implies that
  $$
    r_1\in\left(\frac{e^{\tilde\beta_1\Delta}-1}{\Delta}, 
    \frac{1-\lambda}{1+\lambda}\left(
    \frac{e^{\tilde\beta_2\Delta}-1}{\Delta}\right)\right).
  $$
  % but still need to argue that the case 2 conditions imply Fj(rj)>0
  
  We show uniqueness for Case 2 by contrapositive, which 
  will rule out the possibility of finding solutions for which  
  $F_1(r_1)\leq 0$.  Suppose that
  there exist strictly positive solutions 
  $r_1, r_2$ such that $F_1(r_1)\leq 0$.  Since $F_1(r_1)\leq 0$ if and 
  only if $F_2(r_2)\geq 0$, $r_1\Delta\leq e^{\tilde\beta_1\Delta}-1$, 
  and $r_2\Delta\geq e^{\tilde\beta_2\Delta}-1$, 
  we have that
  $$
    \frac{1+\lambda}{1-\lambda} \geq \frac{r_2}{r_1} 
    \geq \frac{e^{\tilde\beta_2\Delta}-1}{e^{\tilde\beta_1\Delta}-1},
  $$
  as desired. 
\end{proof}
% Now that we've laid the groundwork...
\begin{proof}[Proof of Theorem \ref{thm:discrete_equilibrium}]
  By Theorem \ref{thm:discrete_individual} and Definition \ref{def:discrete_equilibrium}, we must solve \eqref{eqn:clearing_cond_const} and \eqref{eqn:shadow_constraint_const} for the equilibrium shadow interest rates.  Proposition \ref{prop:discrete_F} provides us with the existence and uniqueness of positive shadow interest rates, as desired.
  
  The agents choose not to trade in Case 1, and the market interest rate 
  cannot be uniquely determined.  The annuity values $A$ consistent with 
  this equilibrium must satisfy
  $$
    A \in \left[\frac{A_1}{1+\lambda},\frac{A_1}{1-\lambda}\right]
    \cap \left[\frac{A_2}{1+\lambda},\frac{A_2}{1-\lambda}\right]
    = \left[\frac{\max(A_1,A_2)}{1+\lambda},
    \frac{\min(A_1,A_2)}{1-\lambda}\right].
  $$
  Since $A_i = \frac{\Delta}{e^{\tilde\beta_i\Delta}-1}$ for $i=1,2$, we can rewrite the 
  above interval as
  $$
    A \in 
      \left[\frac{\Delta}{(1+\lambda)\left(e^{\min(\tilde\beta_1, \tilde\beta_2)\Delta}-1\right)}, 
      \frac{\Delta}{(1-\lambda)\left(e^{\max(\tilde\beta_1,
      \tilde\beta_2)\Delta}-1\right)}\right].
  $$
  This interval is nonempty by \eqref{eqn:discrete_case1}.  Since $A=1/r$, we 
  have that 
  $$
    r \in\left[\frac{1-\lambda}{\Delta}
    \left(e^{\max(\tilde\beta_1, \tilde\beta_2)\Delta}-1\right), 
    \frac{1+\lambda}{\Delta}\left(e^{\min(\tilde\beta_1,
    \tilde\beta_2)\Delta}-1\right)\right] \neq \emptyset.
  $$
  
  Trading occurs in Case 2, in which case we are able to determine a unique 
  market interest rate.  When 
  $$
    \frac{e^{\tilde\beta_2\Delta}-1}{e^{\tilde\beta_1\Delta}-1} 
    > \frac{1+\lambda}{1-\lambda},
  $$
  we have that $r_2 = \frac{1+\lambda}{1-\lambda}\cdot r_1$ while $r_1>0$ 
  solves \eqref{eqn:discrete_F_rewritten}.  In this case, $1/r_1=A_1=A(1+\lambda)=(1+\lambda)/r$, which implies that $r = (1+\lambda)r_1$.  Similarly, $r=(1-\lambda)r_2$. Therefore, the market interest rate $r>0$ is determined by
  $$ % i=1, j=2
    \left(1+\frac{r\Delta}{1-\lambda}\right)^{%\frac{1-\lambda}{1+\lambda}\cdot
    \frac{1}{\alpha_2\Delta}}\left(1+\frac{r\Delta}{1+\lambda}\right)^{\frac{1}{\alpha_1\Delta}} 
    = e^{\left(\frac{\tilde\beta_1}{\alpha_1}
    +%\frac{1-\lambda}{1+\lambda}\cdot
    \frac{\tilde\beta_2}{\alpha_2}\right)},
  $$
  and the shadow interest rates are given in terms of
  $$
    r = r_1 (1+\lambda) = r_2 (1-\lambda).
  $$
\end{proof}

In continuous time, the proofs of Theorems~\ref{thm:cts_individual_agent} and \ref{thm:cts_eq} mirror their discrete-time counterparts.  The continuous-time analog of $F_i$ defined in \eqref{def:discrete_F} is given for $i=1,2$ by
$$
  F_i(r):= \frac{1}{\alpha_i}\left(1-\frac{\tilde\beta_i}{r}\right), 
  \ \ \ r>0.
$$
We now prove Theorem~\ref{thm:discrete_to_cts}.
\begin{proof}[Proof of Theorem~\ref{thm:discrete_to_cts}]
  Since $\tilde\beta_2/\tilde\beta_1>\frac{1+\lambda}{1-\lambda}$, 
  Theorem \ref{thm:cts_eq} shows that trade occurs for 
  $\Delta=0$ and $r(0)$ is given uniquely by \eqref{eqn:cts_rate}. 
   Moreover, 
  $$
    \frac{e^{\tilde\beta_2\Delta}-1}{e^{\tilde\beta_1\Delta}-1} 
    \longrightarrow \frac{\tilde\beta_2}{\tilde\beta_1} 
    \ \ \ \text{ as $\Delta\rightarrow 0$,}
  $$
  which by Theorem~\ref{thm:discrete_equilibrium} implies that trade occurs for 
  sufficiently small $\Delta>0$.  In this case, $r(\Delta)>0$ is given 
  uniquely by the solution to \eqref{eqn:discrete_r}.
  
  For $(\Delta, r)\in[0,\infty)\times(0,\infty)$, we define 
  %the function $G$ by 
  $$
   G(\Delta,r):=
   \begin{cases}
     \left(1+\frac{r\Delta}{1+\lambda}\right)^{\frac{1}{\alpha_1\Delta}}
     \left(1+\frac{r\Delta}{1-\lambda}\right)^{\frac{1}{\alpha_2\Delta}}, 
     & \text{for $\Delta>0$}
     \\
     \exp\left(r\left(\frac{1}{\alpha_1(1+\lambda)}
     +\frac{1}{\alpha_2(1-\lambda)}\right)\right), 
     & \text{for $\Delta=0$.}
   \end{cases}
   $$
   For sufficiently small $\Delta>0$ and $\Delta=0$, 
   $r(\Delta)$ is chosen such that 
   $G(\Delta,r(\Delta)) = \exp\left(\frac{\tilde\beta_1}{\alpha_1}+
   \frac{\tilde\beta_2}{\alpha_2}\right)$. 
   Since $G$ is smooth on $[0,\infty)\times(0,\infty)$ and 
   $\frac{\partial G}{\partial r}(0,r(0)) \neq 0$ (a one-sided derivative), 
   the implicit function 
   theorem implies that $r(\Delta)\longrightarrow r(0)$ as 
   $\Delta\rightarrow 0$.
\end{proof}

We now prove Corollary~\ref{cor:small_lambda}.
\begin{proof}
  By Theorem~\ref{thm:cts_eq}, an equilibrium with trade will occur for 
  $\lambda\in\left[0,
  \frac{\tilde\beta_2-\tilde\beta_1}{\tilde\beta_2+\tilde\beta_1}\right)$, 
  in which agent 1 buys shares of the annuity, agent 2 sells shares 
  of the annuity, and the equilibrium interest rate is given by 
  $$
    r(\lambda) = \frac{\tilde\beta_1/\alpha_1+\tilde\beta_2/\alpha_2}{\frac{1}{\alpha_1(1+\lambda)}+\frac{1}{\alpha_2(1-\lambda)}}.
  $$
  Differentiating in $\lambda$, we see that $r$ has local extrema at 
  $$
    \lambda^-= \frac{\left(\sqrt{\alpha_1}-\sqrt{\alpha_2}\right)^2}{\alpha_2-\alpha_1} 
    \ \ \text{ and } \ \ 
    \lambda^+= \frac{\left(\sqrt{\alpha_1}+\sqrt{\alpha_2}\right)^2}{\alpha_2-\alpha_1}.
  $$
  When $\alpha_1>\alpha_2$, we have that both $\lambda^-, \lambda^+<0$.  In this case, $r$ is strictly decreasing on $\left[0,
  \frac{\tilde\beta_2-\tilde\beta_1}{\tilde\beta_2+\tilde\beta_1}\right)$.
  When $\alpha_1<\alpha_2$, we have that $\lambda^-\in(0,1)$, $\lambda^+>1$,  $r$ is strictly increasing on $\left[0,\min\left(\lambda^-,\frac{\tilde\beta_2-\tilde\beta_1}{\tilde\beta_2+\tilde\beta_1}\right)\right)$, and $r$ is strictly decreasing on $\left(\min\left(\lambda^-,\frac{\tilde\beta_2-\tilde\beta_1}{\tilde\beta_2+\tilde\beta_1}\right), \frac{\tilde\beta_2-\tilde\beta_1}{\tilde\beta_2+\tilde\beta_1}\right)$.  Recognizing that $\hat\lambda=\lambda^-$ yields the desired result.
\end{proof}

We next prove Proposition~\ref{prop:ce}.
\begin{proof}
  The sum of the incomplete market certainty equivalents is differentiable on $[0,\frac{\tilde\beta_2-\tilde\beta_1}{\tilde\beta_1+\tilde\beta_2})$.  A calculation of the derivative yields
  $$
    CE_1'(\lambda)+CE_2'(\lambda) = \frac{-2(\tilde\beta_2(1-\lambda)-\tilde\beta_1(1+\lambda))(\alpha_1(1+\lambda)^2+\alpha_2(1-\lambda)^2)}{(\alpha_1(1+\lambda)+\alpha_2(1-\lambda))(1+\lambda)^2(1-\lambda)^2(\tilde\beta_1\alpha_2+\tilde\beta_2\alpha_1)}.
  $$
  Since $\tilde\beta_2(1-\lambda)-\tilde\beta_1(1+\lambda)$ is strictly positive for $\lambda\in[0,\frac{\tilde\beta_2-\tilde\beta_1}{\tilde\beta_1+\tilde\beta_2})$, we have that $CE_1+CE_2$ is strictly decreasing.  On $[\frac{\tilde\beta_2-\tilde\beta_1}{\tilde\beta_1+\tilde\beta_2},1]$, we have that $CE_1+CE_2$ is constant.
  
  We next verify that $CE_1(\lambda)+CE_2(\lambda)<CE_1^{comp}(\gamma)+CE_2^{comp}(\gamma)$, and we recall that the right hand side does not depend on $\gamma$.  Since $\lambda\mapsto CE_1(\lambda)+CE_2(\lambda)$ is decreasing, it suffices to check that the inequality holds for $\lambda = 0$.
  
  By algebra, we have that $CE_1(\lambda)+CE_2(\lambda)<CE_1^{comp}(\gamma)+CE_2^{comp}(\gamma)$ if and only if $\rcmp > r(0)$, which in turn holds if and only if
  $$
    \frac{\alpha_1}{\alpha_2}\sigma_1^2+\frac{\alpha_2}{\alpha_1}\sigma_2^2-2\rho\sigma_1\sigma_2 >0.
  $$
  When $\sigma_1\sigma_2\geq 0$, this inequality holds since
  $$
    \frac{\alpha_1}{\alpha_2}\sigma_1^2+\frac{\alpha_2}{\alpha_1}\sigma_2^2-2\rho\sigma_1\sigma_2
    \geq \frac{\alpha_1}{\alpha_2}\sigma_1^2+\frac{\alpha_2}{\alpha_1}\sigma_2^2-2\sigma_1\sigma_2 
    = \left(\sqrt{\frac{\alpha_1}{\alpha_2}}\sigma_1-\sqrt{\frac{\alpha_2}{\alpha_1}}\sigma_2\right)^2>0.
  $$
  The case when $\sigma_1\sigma_2<0$ is handled analogously.
\end{proof}

Finally, we show Theorem~\ref{thm:bank}.
\begin{proof}[Proof of Theorem~\ref{thm:bank}]
  Assume that $r_1, r_2$ are strictly positive constants. 
  By modifying \eqref{eqn:discrete_derive_F} to account for a traded bank 
  account, we arrive at the same form of $F_i$ as in \eqref{def:discrete_F}.  
  By Definition~\ref{def:bank_equilibrium}~\eqref{def:bank_closeness}, 
  $r_1$ and $r_2$ must satisfy \eqref{eqn:clearing_cond_const} and 
  $$
    \left(\frac{1+r_1}{1+r_2}\right)^n 
    \in \left[\frac{1-\lambda}{1+\lambda}, 
    \frac{1+\lambda}{1-\lambda}\right], \ \ \ \text{ for $n\geq 0$},
  $$
  while for each $n\geq 0$, $B_{1 t_n} = B_{2 t_n}\cdot
  \frac{1+\lambda}{1-\lambda}$ if 
  $\hat\theta_{1 t_n}-\hat\theta_{1t_{n-1}}>0$, and 
  $B_{1 t_n} = B_{2 t_n}\cdot\frac{1-\lambda}{1+\lambda}$ if 
  $\hat\theta_{1 t_n}-\hat\theta_{1 t_{n-1}}<0$.
  Since $\lambda\neq 0$, we must have that $r_1=r_2$ and 
  $\hat\theta_{1t_n}-\hat\theta_{1t_{n-1}} = 
  \hat\theta_{2t_n}-\hat\theta_{2t_{n-1}} = 0$ for all $n\geq 1$.  
  Moreover, \eqref{eqn:clearing_cond_const} implies that 
  $F_1(r_1)=F_2(r_2)=0$, and thus
  $$
    \log(1+r_1\Delta) = \tilde\beta_1\Delta 
    = \tilde\beta_2\Delta
    = \log(1+r_2\Delta),
  $$
  as desired.
\end{proof}

%%%%%%%%%%%%%%%%%%%%%%%%%%%% BIBLIOGRAPHY %%%%%%%%%%%%%%%%%%%%%%%%%%%%
%\bibliographystyle{alpha}          % initials + year style
%\bibliographystyle{plain, abbrv}   % FS style
\bibliographystyle{plain}
\bibliography{finance_bib}

\end{document}